\documentclass[runningheads]{llncs}
\usepackage[T1]{fontenc}
\usepackage{graphicx}
\usepackage{mathtools}
\usepackage{amsmath}
\usepackage{amssymb}
\usepackage{float}
\usepackage{hyperref}
\usepackage{soul}
\usepackage{xcolor}

\newcommand{\resp}[1]{\ (resp. #1)}

\usepackage{xifthen}
\newcommand{\ifnv}[2]{\ifthenelse{\equal{#1}{}}{}{#2}}
\newcommand\sett[3][]{\left\{\left.#2\ifnv{#1}{\in #1}\vphantom{#3}\right|#3\right\}}

\newcommand\dc{\mathfrak d_C}

\newcommand\diam\theta

\newcommand{\N}{\mathbb N}

\newcommand{\Z}{\mathbb Z}

\newcommand\F{\mathbb F}

\usepackage{stmaryrd}

\newcommand{\cc}[2]{\left\llbracket #1,#2\right\rrbracket}

\newcommand{\Cmon}[2]{\mathcal C(#1,#2)}

\begin{document}

\title{Structural Properties of Non-Linear Cellular Automata: Permutivity, Surjectivity and Reversibility}
\titlerunning{Structural Properties of Non-Linear Cellular Automata}
\author{Firas Ben Ramdhane\inst{1,3} \and
Alberto Dennunzio\inst{1} \and
Luciano Margara\inst{2} \and
Giuliamaria Menara\inst{1}}
\authorrunning{F. Ben Ramdhane, A. Dennunzio, L. Margara and G. Menara.}
\institute{Department of Informatics, Systems and Communications,\\
University of Milano-Bicocca, Italy\\
\email{\{firas.benramdhane,alberto.dennunzio,giuliamaria.menara\}@unimib.it}  \and
Department of Computer Science and Engineering,\\ University of Bologna, Cesena, Italy\\ \email{luciano.margara@unibo.it} \\ \and
University of Sfax, Faculty of Sciences of Sfax, Tunisia.\\
}
\maketitle

\begin{sloppypar}
\begin{abstract}
    This paper explores the algebraic conditions under which a cellular automaton with a certain class of non-linear local rules exhibits surjectivity and reversibility. We also analyze the role of permutivity as a key factor influencing these properties and provide conditions that determine whether a non-linear CA in such class is (bi)permutive. In particular, over a prime field, we show that a left-right separated CA is reversible if and only if its local rule depends on a single variable through a permutation, and that a CA whose local rule is a sum of maps of one variable is surjective if and only if one of its two extremal maps is a permutation. Through theoretical results and illustrative examples, we characterize the relationships between these fundamental properties, offering new insights into the dynamical behavior of non-linear CA.
\end{abstract}

\section{Introduction}
\label{sec:intro}

A cellular automaton (CA) is a discrete dynamical system in which each cell updates its state according to local rules, allowing complex global behavior to emerge from simple interactions. CA have been widely employed to model intricate phenomena across diverse scientific fields, including physics~\cite{denby1988neural}, biology~\cite{ermentrout1993cellular}, sociology~\cite{hegselmann1996understanding}, and ecology~\cite{hogeweg1988cellular}. Their conceptual simplicity and modeling flexibility have also attracted considerable interest in computer science, particularly in the domain of \emph{cryptography} (see~\cite{manzoni2025combinatorial} for a comprehensive survey of cryptographic applications).

Among the different classes of CA, linear~\cite{dennunzio2020chaos,dennunzio2021efficiently,dennunzio2024efficient} and additive~\cite{dennunzio2024easy,dennunzio2020dynamical} CA have received considerable attention due to their well-understood algebraic structure and predictable behavior. In contrast, non-linear CA remain much less explored, although some attempts have been made to study both qualitatively and quantitatively the characteristics of such CA~\cite{langton1990computation,wuensche1994complexity,wuensche1999classifying}: this lack of characterization presents both a challenge and an opportunity.

From a theoretical perspective, studying non-linear CA is compelling, as their non-linearity introduces a level of dynamical complexity not present in their linear counterparts. This complexity opens new avenues for analysis and classification, and may reveal behaviors that are fundamentally different from those observed in well-studied classes.

In addition, this complexity and unpredictability make non-linear CA promising candidates for applications where such properties are desirable - most notably in cryptography: while linear CA have already been employed in the construction of various cryptographic primitives, the potential of non-linear CA in this domain remains largely untapped.

The goal of this paper is to take a step toward bridging this gap by initiating a deeper theoretical study of non-linear CA, starting from classical results addressing the injectivity and
surjectivity questions.
It is widely acknowledged that characterizing local rules which make a CA injective or surjective proves arduous in the unrestricted case~\cite{kari2000linear}; therefore, given the complexity of the issue at hand, we limit our analysis to the class of non-linear $j$-separated CA, i.e. CA with diameter $d$ and local rule $f$ defined as:
\[
f(x_1,\dots,x_{d+1}) = a_j x_j^{q_j} + \pi(x_{1}, \dots, x_{j-1},x_{j+1},\dots,x_{d+1}),
\]
where $a_j \in \Z_{m}^*$ is a unit (that is, $a_j$ has a multiplicative inverse in $\Z_m$; equivalently, $\gcd(a_j,m)=1$), $q_j$ is a positive integer, and $\pi:\Z_{m}^{d}\to \Z_{m}$ is any map (for more details see Definition~\ref{def:j-separated}).

Exploiting the structural properties of $j$-separated CA, we are able to provide a permutivity characterization in Lemma~\ref{lem:permutive CA iff (n,phi(m)) coprime} and Proposition~\ref{prop:permutive CA iff (p(x),x^p-x) coprime}.
Then, building on these results and restricting to the class of $LR$-separated CA, where the local rule $f$ can be written, for $1\le \ell< r\le d+1$, as
\[
f(x_1,\dots,x_{d+1}) = a_\ell x_\ell^{q_\ell} + \pi(x_{\ell+1}, \dots, x_{r-1}) + a_r x_r^{q_r},
\]
with units $a_\ell,a_r$, where $x_\ell$ and $x_r$ are the leftmost and rightmost essential variables, that is, variables on which $f$ really depends (or $f=a_\ell x_\ell^{q_\ell}+c$ when $\ell=r$; such CA are called \emph{shift-like}),
we first prove in Theorem~\ref{thm:characterization surjective CA} a sufficient condition under which an $LR$-separated cellular automaton $F$ is surjective, and then in Theorem~\ref{thm:characterization injective/bijective} that $F$ is reversible if and only if it is reversible shift-like, when $m$ is prime or $\pi$ is a polynomial map.

\begin{theorem*}[cf. Theorems~\ref{thm:characterization surjective CA} and~\ref{thm:characterization injective/bijective}]
    Let $F$ be a $LR$-separated CA over the finite ring $\Z_{m}$, for any integer $m \geq 3$, and let $\ell$ \resp{$r$} be the leftmost \resp{rightmost} position of $F$. For an integer $q\ge1$, let $\Cmon{m}{q}$ be the condition: [$q=1$], or [$m$ is squarefree and $\gcd(q,\varphi(m))=1$].
    Then:
    \begin{enumerate}
        \item  If either $\Cmon{m}{q_\ell}$ or $\Cmon{m}{q_r}$ holds, then $F$ is surjective.
        \item If $m$ is prime or the middle map $\pi$ is a polynomial map, then $F$ is injective if and only if $\ell=r$ and $\Cmon{m}{q_\ell}$ holds.
    \end{enumerate}
\end{theorem*}

For comparison, over a prime field a linear CA is reversible exactly when its local rule is a non-zero multiple of a single variable; item~2 extends this fact to $LR$-separated rules. Over $\Z_m$ with $m$ composite, reversible linear CA can have several non-zero coefficients~\cite{ito1983linear,manzini1999complete}; Example~\ref{ex:outside} shows that the unit hypothesis on the extremal coefficients cannot be removed.

We also show that, over $\Z_p$, the condition of item~1 is necessary when the local rule is a sum of maps of one variable (Proposition~\ref{prop:total-surj}); in particular, Conjecture~\ref{conj:characterization surjective CA Zp} holds for all $LR$-separated rules with $r-\ell\le2$, including radius-one rules. The main tool for item~2 is Proposition~\ref{prop:endpoint}, which combines the balance property of surjective CA, a finite Fourier argument and a minimal window of the inverse CA. Theorem~\ref{thm:characterization injective/bijective} also gives conditions on the extremal maps which cover some non-polynomial middle maps over composite rings, and Proposition~\ref{prop:sum} bounds the period of two periodic configurations with the same image when the local rule is a sum of maps of one variable.

Besides theoretical results, we also provide illustrative examples to better clarify the relationships between these fundamental properties, offering new insights into the dynamical behavior of non-linear CA.

\subsection*{Outline}
The paper is organized as follows.
In Section~\ref{sec:background} we recall some relevant algebraic background and notions about the dynamical properties of CA, and we introduce the class of \emph{$j$-separated} non-linear CA, which will be essential in the remainder of this work.
Further, in Section~\ref{sec:permutivity} we proceed by exhibiting algebraic conditions for the local rule $f$ under which the CA is leftmost (or rightmost) permutive.
Finally, in Sections~\ref{sec:surjectivity global rule} and~\ref{sec:reversibility} we present the core results of the paper, providing characterization theorems for surjective and reversible $j$-separated CA.

\subsection*{Acknowledgements}
This work was partially supported by the PRIN 2022 PNRR project ‘‘Cellular Automata Synthesis for Cryptography Applications (CASCA)’’ (P2022MPFRT) funded by the European Union – Next Generation EU, and by the
HORIZON-MSCA-2022-SE-01 project 101131549 ‘‘Application-driven Challenges for Automata Networks and Complex Systems (ACANCOS)’’.

\section{Terminology and Background}
\label{sec:background}

In this section, we give the preliminary definitions and results needed for the rest of the paper. 
For a comprehensive introduction on the theory of CA see~\cite[Section 1]{back2012handbook} and \cite[Chapter 5]{kurka2003topological}.

We start with some terminology from word combinatorics. An \emph{alphabet} 
$A$ is a finite set of symbols, called \emph{letters}. In this paper, we take $A=\Z_m$, the set of integers modulo $m$.
A \emph{finite word} over an alphabet $A$ is a finite sequence of letters from $A$. The length of a finite word $u$, denoted by $|u|$, is the number of letters it contains. The unique word of length $0$ is called the \emph{empty word} and is denoted by $\lambda$.
A \emph{configuration} (or \emph{bi-infinite word}) $x = \ldots x_{-2}x_{-1}x_0x_1x_2 \ldots$ over $A$ is an infinite concatenation of letters from $A$ indexed by $\Z$. For integers $n \le m$, we denote by $x_{\cc{n}{m}} = x_n x_{n+1} \cdots x_{m-1}x_m$ the subword of $x$ from position $n$ to $m$, where $\cc{n}{m} = [n,m] \cap \Z$; further, we will indicate by $u^\infty$ the \emph{constant word}, i.e. the word constructed by concatenating the same letter $u$ infinitely many times; more generally, for a non-empty finite word $w$, $w^\infty$ denotes the configuration obtained by repeating $w$ in both directions, with the first letter of $w$ at position $0$.
The set of all finite (resp. bi-infinite) words over $A$ is denoted by $A^*$ (resp. $A^\Z$), and for each $n \in \N$, the set of words of length $n$ is denoted by $A^n$.
Most classically, the set $A^\Z$ is endowed with the product topology of the discrete topology on each copy of $A$. The topology defined on $A^\Z$ is metrizable, corresponding to the \emph{Cantor distance} defined as follows:
$$ \dc(x,y)=2^{-\min\sett{|n|}{x_n\neq y_n,~ n\in\Z}}, \forall x\neq y\in A^\Z, \text{ and } \dc(x,x)=0, \forall x\in A^\Z.$$
This space, called the \emph{Cantor} space, is compact, totally disconnected and perfect. 

This topological framework naturally leads to the definition of a \emph{topological dynamical system}, which provides a formal setting for studying the evolution of configurations under continuous transformations. Recall that a \emph{topological dynamical system} is a pair $(X_d, F)$, where $X_d = (X, d)$ is a compact metric space and $F \colon X \to X$ is a continuous map. When $X$ consists of symbolic configurations, such as elements of $A^\Z$, the system is called a \emph{symbolic dynamical system}.

CA are a classical example of such systems. Formally, a CA is a map $F \colon A^{\mathbb{Z}} \to A^{\mathbb{Z}}$ such that there exist an integer $\rho \ge 0$ and a local rule $f \colon A^{2\rho+1} \to A$ satisfying, for all $x \in A^{\mathbb{Z}}$ and $i\in\Z:$
$F(x)_i = f(x_{\cc{i-\rho}{i+\rho}}).$
We refer to $\rho$ as the \emph{radius} and $d = 2\rho$ as the \emph{diameter} of the CA.
The variables of $f$ are numbered $1,\dots,d+1$, and variable $j$ is read at the offset $j-\rho-1$. A variable is \emph{essential} if changing it alone can change the value of $f$; we denote by $\ell$ and $r$ the positions of the leftmost and rightmost essential variables. If we delete the inessential variables outside $\cc{\ell}{r}$ but keep the offsets of the other ones, we obtain exactly the same global map; in particular $r-\ell\le d$. Changing all offsets by the same integer composes the global map with a power of the shift $\sigma$ defined below, which changes neither injectivity nor surjectivity.
A fundamental example of a CA is the \emph{shift map}, defined on $A^{\mathbb{Z}}$ by $\sigma(x)_i = x_{i+1}$ for all $i \in \mathbb{Z}$. This map plays a central role in the theory of symbolic dynamics, particularly in the characterization of CA. In fact, a classical result by Curtis, Hedlund, and Lyndon~\cite{hedlund1969endomorphisms} states that a function $F \colon A^{\mathbb{Z}} \to A^{\mathbb{Z}}$ is a CA if and only if it is continuous (with respect to the product topology) and commutes with the shift, that is, $F(\sigma(x)) = \sigma(F(x))$ for all $x \in A^{\mathbb{Z}}$.
Another result by Hedlund characterizes surjective CA. Recall that a CA is said to be \emph{surjective} (resp. \emph{injective}) if its global rule $F$ is onto (resp. one-to-one) and \emph{bijective} if $F$ is both onto and one-to-one. 

To state this result, we define the extension of the local rule $f$, denoted by $f^*$, of a CA $F$ with diameter $d$, on $A^*$ as follows: 
$f^*(u)_i=f(u_{\cc{i}{i+d}})$ if $i<|u|-d$ and the empty word otherwise.
\begin{theorem}[\cite{hedlund1969endomorphisms}]\label{theo: hedlund surj}
A CA $F$ with local rule $f$ and diameter $d$ is surjective if and only if for all $u\in A^*\setminus \{\lambda\}$, $\#{f^*}^{-1}(u)=(\#A)^{d}.$ 
\end{theorem}


A well-known result states that every injective CA is also surjective \cite[Corollary 5.27]{kurka2003topological}. As a consequence, a CA is bijective if and only if it is injective. Moreover, the inverse of a bijective CA is itself a CA and thus a CA $F$ is injective if and only if it is \emph{reversible}, i.e., there is a CA $G$ such that $F\circ G=G\circ F=id$, where $id$ is
the identity function, for more details one can see \cite{kari2005reversible}.

A distinct and particularly relevant class of CA also ensures surjectivity: these are the so-called \emph{permutive} CA.
We say that a CA $F$ of diameter $d$ and local rule $f$ is \emph{permutive at position} $i$ (with $1 \le i \le d+1$) if, for every $u \in A^{i-1}$, every $v \in A^{d-i+1}$, and every $b \in A$, there exists a unique $a \in A$ such that $f(uav) = b$. In other words, when all variables except the $i$-th are fixed, the function $f$ acts as a permutation in the $i$-th variable.
In particular, if $i = 1$ (respectively, $i = d+1$), we say that $F$ is \emph{left} (respectively, \emph{right}) permutive. A CA is said to be \emph{bipermutive} if it is both left and right permutive, and simply \emph{permutive} if it satisfies at least one of these conditions.
According to \cite[Proposition 5.22]{kurka2003topological}, every permutive CA is surjective.

We now turn our attention to an algebraic notion and a result which we will rely on in the upcoming results.
Recall that the \emph{Euler's totient function}~\cite{euler2012theoremata}, denoted $\varphi(n)$, is defined as the number of positive integers less than or equal to $n$ that are coprime to $n$. Formally,
\[
\varphi(n)= \# \{k \in \Z \text{ such that } 1 \leq k \leq n \text{ and } \gcd(k,n)=1 \}.
\]
Also, recall that every function from a finite field $\F$ to itself can be represented as a polynomial over $\F$. For completeness, we include a proof of this fact.

\begin{lemma}
\label{prop:any function is poly}
Any function from a finite field $\F$ to $\F$ can be represented as a polynomial over $\F$.
\end{lemma}

\begin{proof}
    Let $g$ be any function from $\F$ to $\F$. Since $\F$ is a finite field, then we have a finite list of all the elements of $\F$, label them $h_1, h_2, \dots,h_n$, where the $h_i$ are pairwise distinct.
    We want to write $g$ as a polynomial $g(x)=a_0 + a_1x+\dots+a_{n-1}x^{n-1}$ of degree at most $n-1$ that has the value $g(h_i)$ at $h_i$ for every $i$.
    This is equivalent to solving the following system:
    \[
    \begin{matrix}
        a_0 + a_1(h_1) + a_2(h_1)^2 + \cdots + a_{n-1}(h_1)^{n-1} = g(h_1) \\
         \dots\dots\dots\dots\dots \\
        a_0 + a_1(h_n) + a_2(h_n)^2 + \cdots + a_{n-1}(h_n)^{n-1} = g(h_n). \\
    \end{matrix}
    \]

    This linear system can be represented with the help of the Vandermonde matrix, $V$:
    \[
    V=\left[
    \begin{matrix}
        1 &h_1 &h_1^2 &\cdots &h_1^{n-1} \\
        &\dots&\dots&\dots& \\
        1 &h_n &h_n^2 &\cdots &h_n^{n-1} \\
    \end{matrix}
    \right].
    \]
    Now if we let $a = [a_0, a_1,\dots, a_{n-1}]^T$ and let $g = [g(h_1), g(h_2),\dots, g(h_n)]^T$, then we can find our coefficients by solving the system $V a = g$ for $a$.
    The question of whether the polynomial exists reduces to the question of whether the determinant of $V$ is non-zero. 
    The determinant of the Vandermonde matrix is $\prod_{i < j}(h_j-h_i)$,~\cite{horn2012matrix}. Since the $h_i$ are pairwise distinct, the determinant is non-zero and therefore a unique solution exists.
    So for every function from $\F$ to $\F$, we can find a polynomial over $\F$ that agrees with that function at every point. \qed
\end{proof}

\begin{remark}
\label{rem:any function is poly on Fn}
    If $\F$ is a finite field, it is straightforward to extend Lemma \ref{prop:any function is poly} to the context of a function $f:\F^n \to \F$.
\end{remark}

\begin{remark}
\label{rem:function over finite ring not necessarily poly}
    This result fails to extend to any $\Z_m$ where $m$ is not prime.
    Indeed, consider the Kronecker delta function on $\Z_4$,
    \[
    f: \Z_4 \to \Z_4, \quad f(x) = 
    \begin{cases}
    1 & \text{if } x = 0 \\
    0 & \text{otherwise}.
    \end{cases}
    \]
    Suppose there exists a polynomial $p(x) \in \Z_4[x]$ such that:
    \[
    p(0) = 1,\quad p(1) = p(2) = p(3) = 0
    \]
    Consider the polynomial $q(x) = p(x) - 1$. Then:
    \[
    q(0) = 0,\quad q(1) = q(2) = q(3) = -1 \equiv 3 \pmod{4}.
    \]
Since $q(0) = 0$, $x$ must divide $q(x)$ in $\Z_4[x]$. Thus, we can write:
    \[
    q(x) = x \cdot r(x),
    \]
    for some polynomial $r(x) \in \Z_4[x]$.

    Evaluating $q(x)$ at $x = 2$:
    \[
    q(2) = 2 \cdot r(2) \equiv 3 \pmod{4}.
    \]
    However, $2 \cdot r(2)$ can only be $0$ or $2$ modulo 4, never $3$, so this is a contradiction.
\end{remark}

While we acknowledge, from Remark~\ref{rem:function over finite ring not necessarily poly}, that restricting our focus to polynomial local rules over $\Z_m$ will result in an ultimately incomplete analysis, we choose to start by studying this simpler case, as it will serve as a foundation allowing us to later address the more general scenario.
However, it is worth pointing out that, because of Lemma~\ref{prop:any function is poly} and Remark \ref{rem:any function is poly on Fn}, when $m$ is prime (i.e., in the context of finite fields) the investigation of polynomial local rules amounts to an exhaustive analysis of non-linear CA.
\\

Although we are restricting our focus to polynomial functions, the study remains complex due to the wide range of behaviors these functions exhibit. To manage this complexity, we further narrow our attention to a specific class of non-linear CA defined by a local rule $f$ such that $f$ is a multivariate polynomial with (at least) one variable separated from the others. We end this section by introducing this notion, which we will rely on in the remainder of the paper.

\begin{definition}
\label{def:j-separated}
Let $F$ be a CA over the finite ring $\Z_{m}$ with $m \geq 3$, defined by a local rule $f: \Z_m^{d+1} \to \Z_m$ of the form:
$$f(x_1, \dots, x_{d+1}) = a_j x_j^{q_j} + \pi(x_1,\dots, x_{j-1},x_{j+1},\dots,x_{d+1}),$$
where $a_j \in \Z_m^*$ is a unit, $q_j\geq 1$ is an integer, and $\pi: \Z_m^d \to \Z_m$ is any map.
\begin{enumerate}
    \item We say that $F$ is \emph{separated in position $j$}, or simply \emph{$j$-separated}.
    \item If $j = \ell$ (resp.\ $j = r$), where $\ell$ (resp.\ $r$) is the position of the leftmost (resp.\ rightmost) essential variable of $f$, then $F$ is said to be \emph{leftmost} (resp.\ \emph{rightmost}) separated.
    \item We say that $F$ is \emph{$LR$-separated} if it is both leftmost and rightmost separated.
    \item We say that $F$ is \emph{totally separated} if the local rule is of the form
    $$f(x_1,\dots,x_{d+1}) = \sum_{i=1}^{d+1} a_i x_i^{q_i},$$
    where $a_i\in\Z_m$ and $q_i\geq1$.
\end{enumerate}
\end{definition}

\begin{remark}\label{re: def sep}
If $F$ is a $LR$-separated CA with local rule $f$ and diameter $d$, then $f$ necessarily takes one of the following forms:
\begin{enumerate}
 \item $f(x_1,\dots,x_{d+1}) = a_\ell x_\ell^{q_\ell}+c$ with $c\in\Z_m$, in which case $\ell = r$ and $F$ is said to be \emph{shift-like}.
    \item $f(x_1,\dots,x_{d+1}) = a_\ell x_\ell^{q_\ell} + \pi(x_{\ell+1}, \dots, x_{r-1}) + a_r x_r^{q_r}$,
    where $1 \le \ell < r \le d+1$, the coefficients \(a_\ell\) and \(a_r\) are units, and \(\pi : \Z_m^{r - \ell - 1} \to \Z_m\) is an arbitrary map (a constant if $r=\ell+1$).
\end{enumerate}

The two cases must be written separately: in the first case, the formula of the second case would count the same variable twice. \\
We will refer to $\ell$ (resp. $r$) as the \emph{leftmost} (resp. \emph{rightmost}) position of $F$.
\end{remark}

It is important to note that this work focuses on the case $A = \Z_m$ with $m \geq 3$, as for $m = 2$ every totally separated CA is linear (since $x^q=x$ in $\Z_2$), and linear CA have already been extensively studied in the literature (see, for example, \cite{ito1983linear} and \cite{manzini1999complete}). An $LR$-separated rule over $\Z_2$ can still be non-linear through its middle map, as $x_1+x_2x_3+x_4$ shows.

\section{Quadratic CA on finite fields}

Among non-linear CA, a particularly notable subclass is that of quadratic CA. We begin by proving that no such automaton can be surjective over a finite field $\Z_p$, where $p$ is an odd prime. Although for totally separated rules this result also follows from Proposition~\ref{prop:surjective if PP or non-PP}, we include it here explicitly, as the constructive argument provides valuable insight into the structural constraints specific to this class.


\begin{definition}
Let $p$ be an odd prime. A CA $F$ with diameter $d$ and local rule $f: \Z_p^{d+1}\to\Z_p$ is quadratic if $f$ is a quadratic form on $\Z_p^{d+1}$.
\end{definition}
Recall that $f$ is a quadratic form if $f(au)=a^2f(u)$ for any $u\in \Z_p^{d+1}$ and $a\in \Z_p$, and the map $(u,v)\mapsto f(u+v)-f(u)-f(v)$ is a bilinear form, that is, it is linear in each argument separately.

\begin{lemma}
\label{lem:CA if surj not quadratic}
Let $F$ be a totally separated CA over the finite field $\Z_p$, where $p$ is an odd prime, i.e. the local rule $f$ is given by
$$f(x_1, \dots, x_{d+1}) = \sum_{i=1}^{d+1} a_i x_i^{q_i},$$
where each $a_i \in \Z_p$.
If every $q_i$ is an even positive integer, then the global map $F$ is not surjective.
\end{lemma}

\begin{sloppypar}
\begin{proof}
    Suppose that such CA $F$ is surjective. Then $f$ is not constant, so there exists $h\in\cc{1}{ d+1}$ such that $a_h\neq 0$. Consider the set $S=f^{-1}(a_h)$; it is not empty, since it contains the vector with $1$ at position $h$ and $0$ elsewhere.
    If $v\in S$, then $-v\in S$: changing the sign of a coordinate $x_i$ does not change $x_i^{q_i}$, because $q_i$ is even. Moreover $v\neq -v$, since $p$ is odd and $v\neq0$ (indeed $f(0)=0\neq a_h$). Hence the map $v\mapsto -v$ splits $S$ into pairs, and $\# S=\# f^{-1}(a_h)$ is an even number.
    On the other hand, by Theorem \ref{theo: hedlund surj} and the surjectivity of $F$, we obtain $\# f^{-1}(a_h)=p^{d}$ which contradicts the fact that $p^{d}$ is an odd number (as $p$ is an odd prime number). Therefore, we can conclude that $F$ is not surjective. \qed
\end{proof}
\end{sloppypar}

The proof only uses the equality $f(-u)=f(u)$. More generally, if $p$ is an odd prime and the local rule satisfies $f(-u)=f(u)$ for every $u\in\Z_p^{d+1}$, then $F$ is not surjective. This is clear if $f$ is constant; otherwise, take a value $b\neq f(0)$ of $f$: the map $u\mapsto -u$ splits $f^{-1}(b)$ into pairs, so $\#f^{-1}(b)$ is even and cannot be equal to $p^d$.

We can specialize Lemma~\ref{lem:CA if surj not quadratic} to the context of quadratic local rules, yielding a corresponding result for quadratic CA.

\begin{corollary}
There is no surjective quadratic CA over $\Z_p$ for any prime $p\ge 3$.
\end{corollary}
\begin{proof}
   Let $F$ be a CA with diameter $d$ and quadratic local rule $f$. If $f=0$, then $F$ is not surjective. Otherwise, choose a non-zero value $b$ of $f$. Since $f(-u)=(-1)^2f(u)=f(u)$ and $f(0)=0$, the pairing argument in the proof of Lemma~\ref{lem:CA if surj not quadratic} shows that $\#f^{-1}(b)$ is even, which contradicts Theorem~\ref{theo: hedlund surj}. This argument also covers quadratic forms with mixed terms $x_ix_j$. \qed
\end{proof}  
   
\begin{corollary}
\label{cor:qi all even then not injective} 
    Let $F$ be a totally separated CA over $\Z_p$ for any prime $p\ge 3$.\\
    If the powers $q_i$'s are all even positive integers, then $F$ is not reversible.
\end{corollary}

\begin{remark}
It follows from Lemma ~\ref{lem:CA if surj not quadratic}, together with \cite[Theorems 5.49 and 5.50]{kurka2003topological}, \cite[Proposition 5.41]{kurka2003topological} and \cite{codenotti1996transitive}, that CA satisfying the hypotheses of Lemma~\ref{lem:CA if surj not quadratic} are not open, closing, positively expansive or transitive, where the last two properties are typically linked to chaotic behavior in symbolic dynamical systems. For precise definitions, see \cite[Definitions 2.39 and 1.15]{kurka2003topological}.
\end{remark}

\section{Permutivity}
\label{sec:permutivity}

In this section, we focus on the study of the permutivity property of non-linear $j$-separated CA.

\begin{lemma}\label{lem:permutive CA iff (n,phi(m)) coprime} 
    Take a finite ring $\Z_{m}$, where $m\ge 2$ is an integer.
    Let $F$ be a $j$-separated CA of diameter $d$, so the local rule $f$ can be written as
    \[
    f(x_1, \dots, x_{d+1}) = ax_j^{n} + g(x_1,\dots, x_{j-1},x_{j+1},\dots,x_{d+1}),
    \]
    where $a \in \Z_{m}$ is invertible (that is, $a$ is a unit: $a^{-1}$ exists in $\Z_m$), $n\ge1$, and $g:\Z_{m}^{d} \to \Z_{m}$ is any map.

    Then $F$ is permutive in position $j$ if and only if $\Cmon{m}{n}$ holds, that is, $n=1$, or $m$ is squarefree and $\gcd(n,\varphi(m))=1$, where $\varphi(m)$ is the Euler's totient function of $m$.
\end{lemma}


\begin{remark}
\sloppy
    In particular, if $F$ is $(d+1)$-separated [resp. $1$-separated] then $F$ is right-permutive [resp. left-permutive] if and only if $\Cmon{m}{n}$ holds.
\end{remark}

\begin{proof}
Fix all variables other than $x_j$. Adding the resulting constant and multiplying by the unit $a$ are bijections, so the question is whether $x\mapsto x^n$ is a permutation of $\Z_m$. For $n=1$ it is the identity. Let $n\ge2$ and write $m=\prod_k p_k^{e_k}$. By the Chinese Remainder Theorem, $x\mapsto x^n$ is bijective on $\Z_m$ if and only if it is bijective on every $\Z_{p_k^{e_k}}$. If some $e_k\ge2$, then $p_k^{e_k-1}\neq0$ but $(p_k^{e_k-1})^n=0$ in $\Z_{p_k^{e_k}}$, since $n(e_k-1)\ge e_k$; so the map is not injective. If $e_k=1$, the map fixes $0$ and acts on the cyclic group $\Z_{p_k}^*$ of order $p_k-1$, where it is bijective if and only if $\gcd(n,p_k-1)=1$~\cite{dummit2004abstract}. For squarefree $m$ these conditions together mean $\gcd(n,\prod_k(p_k-1))=\gcd(n,\varphi(m))=1$. \qed
\end{proof}

\begin{remark}
    \label{rem:permutive CA iff (n,p-1) coprime}
    Notice that it follows from Lemma~\ref{lem:permutive CA iff (n,phi(m)) coprime} that if $m$ is a prime number and $\Z_m$ is the finite field with $m$ elements, then $F$ is permutive in position $j$ if and only if $\gcd(n,m-1)=1$, since $\varphi(m)=m-1$ for $m$ prime. The squarefree condition cannot be omitted for general $m$: for example $0^5=3^5=0$ in $\Z_9$, although $\gcd(5,\varphi(9))=1$. On the other hand, $x\mapsto x^3$ is a permutation of $\Z_{15}$, since $15$ is squarefree and $\gcd(3,\varphi(15))=\gcd(3,8)=1$.
\end{remark}

The additive separation of $x_j$ is essential in Lemma~\ref{lem:permutive CA iff (n,phi(m)) coprime}: it is the reason why the condition depends only on the exponent $n$. With mixed terms, every polynomial obtained by fixing the other variables must be tested. For example, $x\mapsto x^3$ permutes $\Z_5$, but $f(x,y)=x^3+xy$ is not permutive in $x$: for $y=1$, the values at $x=0$ and $x=2$ are both $0$.

It was shown by Hermite 
in~\cite{hermite1863sur} (see also~\cite{lidl1997finite}) that a polynomial $f$ over a finite field $\F_p$ is invertible if and only if $f$ has exactly one root in $\F_p$ and for each integer $t$ with $1 \le t \le p-2$, $t \not\equiv 0 \mod p$, the reduction of $[f(x)]^t \mod (x^p-x)$ has degree at most $p-2$.
Therefore, a CA over $\Z_p$ with local rule $f(x_1,\dots,x_{d+1})=\pi(x_{d+1}) + g(x_1,\dots,x_d)$ [resp. $f(x_1,\dots,x_{d+1})=\pi(x_1) + g(x_2,\dots,x_{d+1})$] is right-permutive [resp. left-permutive] if and only if the two aforementioned conditions hold for the polynomial $\pi(x)$.

In the context of the finite field on $p$ elements $\Z_p$, with $p$ prime, the condition $t\not\equiv 0\mod p$ is automatic. Note that a condition on the derivative, such as $\gcd(f'(x),x^p-x)=1$, does not characterize permutation polynomials: over $\Z_5$, $x^3$ is a permutation polynomial although its derivative vanishes at $0$, and $x^3+x$ is not a permutation polynomial (its values at $0$ and $2$ are both $0$) although $3x^2+1$ has no root in $\Z_5$.
We thus have the following characterization of permutive CA over the finite field $\Z_p$.

\begin{proposition}
\label{prop:permutive CA iff (p(x),x^p-x) coprime}
Let $F$ be a CA over the finite field $\Z_p$ with diameter $d$ defined by the local rule $$f(x_1,\dots,x_{d+1})=\pi(x_j) + g(x_1,\dots,x_{j-1},x_{j+1},\dots,x_{d+1}),$$ where $\pi(x) \in \Z_p[x]$ is a polynomial and $g$ is any map $g:\Z_p^{d} \to \Z_p$.
Then $F$ is permutive in position $j$ if and only if $\pi$ has exactly one root in $\Z_p$ and, for every integer $t$ with $1\le t\le p-2$, the remainder of $\pi(x)^t$ modulo $x^p-x$ has degree at most $p-2$.
\end{proposition}


\begin{proof}
For $x_1,\cdots, x_d\in \Z_p$, let $h:\Z_p \to \Z_p$ be the map defined as
    \[
    h(x)=\pi(x) + g(x_1,\dots,x_d).
    \]
Since translation by the constant $g(x_1,\dots,x_d)$ is a bijection of $\Z_p$, the map $h$ is a permutation if and only if $\pi$ is a permutation polynomial, which is characterized by Hermite's criterion above \cite{hermite1863sur,lidl1997finite}. No condition on the degree of $\pi$ is needed, since by Fermat's Little Theorem ($x^p \equiv x \mod p$) the reduction of $\pi$ modulo $x^p-x$ has the same values on $\Z_p$. \qed

		
\end{proof}


\section{Surjectivity}
\label{sec:surjectivity global rule}

We now provide some alternative characterization results on surjectivity for the class of $LR$-separated CA.
We start by recalling some useful facts from~\cite{mullen2013handbook}.

\begin{definition}[{\cite[Def. 8.2.1]{mullen2013handbook}}]
    Let $\F_p$ be the finite field with $p$ elements. 
    A polynomial $f \in \F_p[x_1,...,x_n]$ is a permutation polynomial in $n$ variables over $\F_p$ if the equation $f(x_1,...,x_n) = \alpha$ has exactly $p^{n-1}$ solutions in $\F^n_p$ for each $\alpha \in \F_p$.
\end{definition}

\begin{theorem}[{\cite[Theorem 8.2.9]{mullen2013handbook}}]
\label{thm:mullen}
Let $f \in \F_p[x_1,...,x_n]$ be of the form
\[
f(x_1,...,x_n)=g(x_1,...,x_m)+h(x_{m+1},...,x_n), \hspace{.5cm} 1\leq m <n.
\]
If at least one of $g$ and $h$ is a permutation polynomial over $\F_p$ then $f$ is a permutation
polynomial over $\F_p$. If $p$ is prime, then the converse also holds.  
\end{theorem}

The following is a direct consequence of the results above.

\begin{proposition}
\label{prop:surjective if PP or non-PP}
    Let $F$ be a $LR$-separated CA with local rule $f$ over $\Z_p$, for any prime $p\ge 3$ and let $\ell$ \resp{$r$} be the leftmost \resp{rightmost} position of $F$.
    \begin{enumerate}
        \item If the polynomial $\pi$ (defined as in Remark \ref{re: def sep}) is any non-permutation polynomial, then $F$ is surjective if and only if $\gcd(q_{\ell},p-1)=1$ or $\gcd(q_r,p-1)=1$.
        \item If $F$ is a totally separated surjective CA, then there is at least one $j \in \cc{\ell}{r}$ such that $a_j\neq0$ and $\gcd(q_j,p-1)=1$.
    \end{enumerate}
   
\end{proposition}
The following states our main result regarding surjectivity for $LR$-separated CA.

\begin{theorem}
\label{thm:characterization surjective CA}
    Let $F$ be a $LR$-separated CA with local rule $f$ over the finite ring $\Z_{m}$, for any integer $m \geq 3$, and let $\ell$ \resp{$r$} be the leftmost \resp{rightmost} position of $F$. 
     If either $\Cmon{m}{q_\ell}$ or $\Cmon{m}{q_r}$ holds, then $F$ is surjective.
\end{theorem}

\begin{sloppypar}
\begin{proof}
 Suppose without loss of generality that $\Cmon{m}{q_r}$ holds (in case $\Cmon{m}{q_\ell}$ holds the proof is similar). If $\ell=r$, then $F$ applies the permutation $t\mapsto a_\ell t^{q_\ell}+c$ (Lemma~\ref{lem:permutive CA iff (n,phi(m)) coprime}) at the position of $x_\ell$, so $F$ is bijective. Assume now $\ell<r$.
 Since $F$ is $LR$-separated, its local rule can be written in the form:
     \[
    f(x_{1},\dots,x_{d+1})=a_{\ell}x_{\ell}^{q_\ell} + \pi(x_{\ell +1},\dots,x_{r-1}) + a_{r}x_{r}^{q_r},
    \]
    where $d$ is the diameter of $F$. We then consider the cellular automaton $G$ whose local rule is obtained by restricting $f$ to its essential neighborhood, i.e., the map
    $g:\Z_{m}^{r-\ell +1} \to \Z_{m}$ defined as
    \[
    g(x_{\ell},\dots,x_r)=a_{\ell}x_{\ell}^{q_\ell} + \pi(x_{\ell +1},\dots,x_{r-1}) + a_{r}x_{r}^{q_r},
    \]
    where each variable keeps the offset it has in $f$.
    Note that $F$ and $G$ induce the same global map, i.e., $F(x)=G(x)$ for every configuration $x$. This holds for all $\ell<r$, even when $r-\ell$ is large: we only delete inessential variables and we do not move the other ones, so no condition on the radius is needed.
    Since $\Cmon{m}{q_r}$ holds, then by Lemma~\ref{lem:permutive CA iff (n,phi(m)) coprime}, $g$ is right-permutive, and thus $G$ is surjective.
    Therefore for every $y \in \Z_{m}^\Z$ there exists $x \in \Z_{m}^\Z$ such that $G(x)=y$.
  Finally, because of the definition of the local rule $f$ we can conclude that $F(x)=y$, meaning $F$ is surjective. \qed
\end{proof}

\begin{remark}
    As in the case of Lemma~\ref{lem:permutive CA iff (n,phi(m)) coprime}, it follows from Theorem~\ref{thm:characterization surjective CA} that, when $m$ is prime, $F$ is surjective whenever $\gcd(q_{\ell},m-1)=1$ or $\gcd(q_r,m-1)=1$.
\end{remark}

Notice that the converse implication does not hold in general in ring $\Z_m$. 

\begin{example}
\label{Ex sur non perm}

Let $F$ be the CA over $\Z_4$ with local rule $f$ given by: 
$$f(a,b,c)=a^2+b+c^2 \mod 4, \forall a,b,c\in \Z_4.$$
It is clear that $F$ could not be right (resp. left) permutive, since $0^2=2^2$ in $\Z_4$; however, $F$ is a surjective CA, as verified algorithmically since surjectivity is decidable for one-dimensional cellular automata~\cite{amoroso1972decision}. A direct proof is also short: use the neighborhood $\cc{-1}{1}$ and write $x_i=u_i+2v_i$ with $u_i,v_i\in\{0,1\}$. Since $x_i^2=u_i$ in $\Z_4$, the equation $F(x)=y$ becomes $y_i=u_{i-1}+u_i+u_{i+1}+2v_i \mod 4$. Choose $u_{-1},u_0$ freely; the equation $u_{i-1}+u_i+u_{i+1}=y_i \mod 2$ then determines all the other $u_i$, to the right and to the left, and after that each $v_i\in\{0,1\}$ is uniquely determined. So every configuration has exactly four preimages.
\end{example}
\end{sloppypar}

We conclude this section conjecturing that a full characterization of surjectivity completing the result of Theorem \ref{thm:characterization surjective CA} holds for LR-separated  CA over a finite field $\Z_p$.

\begin{conjecture}
\label{conj:characterization surjective CA Zp}
    Let $F$ be a $LR$-separated CA over the finite field $\Z_{p}$, for any prime integer $p \geq 3$, and let $\ell$ \resp{$r$} be the leftmost \resp{rightmost} position of $F$. 
     Then $F$ is surjective if and only if either $\gcd(q_{\ell},p-1)=1$ or $\gcd(q_r,p-1)=1$.
\end{conjecture}

Proposition~\ref{prop:surjective if PP or non-PP} proves Conjecture~\ref{conj:characterization surjective CA Zp} when $\pi$ is not a permutation polynomial. In the remaining case $\pi$ is balanced, and then so is the local rule $f$. Theorem~\ref{thm:mullen} only uses this balance of $f$, that is, the condition of Theorem~\ref{theo: hedlund surj} for words of length one, and this is not enough: over $\Z_3$, the rule $f(a,b,c)=a^2+b+c^2$ is balanced, since it is permutive in $b$, but the word $22$ has $13$ local preimages instead of $9$, so the CA is not surjective by Theorem~\ref{theo: hedlund surj}. A proof of the conjecture must therefore control words of every length. The next result does this when the local rule is a sum of maps of one variable, as in this example. Its proof uses character sums: for a prime $p$ we write $\omega=e^{2\pi i/p}$, and for $b\in\Z_p$ the power $\omega^b$ is well defined, since $\omega^p=1$.

\begin{lemma}\label{lem:characters}
Let $p$ be prime and $\psi:\Z_p\to\Z_p$. Then $\sum_{t\in\Z_p}\omega^{\psi(t)}=0$ if and only if $\psi$ is a permutation.
\end{lemma}

\begin{proof}
Let $N_b=\#\psi^{-1}(b)$, so that the sum is $P(\omega)$ with $P(X)=\sum_{b=0}^{p-1}N_bX^b$. If $\psi$ is a permutation, then $P(\omega)=1+\omega+\dots+\omega^{p-1}=0$. Conversely, if $P(\omega)=0$, then $P$ is a multiple of $1+X+\dots+X^{p-1}$, the minimal polynomial of $\omega$ over $\mathbb Q$; since $P\neq0$ has degree at most $p-1$, all $N_b$ are equal, hence equal to $1$. \qed
\end{proof}

\begin{proposition}\label{prop:total-surj}
Let $p$ be prime and $D\ge1$, and let $F$ be a CA over $\Z_p$ whose local rule, with neighborhood $\cc{0}{D}$, is
\[
f(x_0,\dots,x_D)=\sum_{j=0}^{D}g_j(x_j),
\]
where $g_0,\dots,g_D:\Z_p\to\Z_p$ are any maps.
\begin{enumerate}
\item If $g_0$ and $g_D$ are not constant and are not permutations, then some word of length at most $p^D-D$ does not have exactly $p^D$ local preimages, so $F$ is not surjective.
\item Consequently, a CA over $\Z_p$ whose local rule is a sum of maps of one variable is surjective if and only if one of its two extremal essential maps is a permutation. In particular, Conjecture~\ref{conj:characterization surjective CA Zp} holds for all totally separated $LR$-separated rules, and for all $LR$-separated rules with $r-\ell\le2$, including radius-one rules.
\end{enumerate}
\end{proposition}

\begin{proof}
We first choose $t\neq t'$ with $g_0(t)\neq g_0(t')$ and $g_D(t)\neq g_D(t')$. Take $a,b$ with $g_0(a)\neq g_0(b)$. If $g_D(a)\neq g_D(b)$, take $(t,t')=(a,b)$. Otherwise take $c$ with $g_D(c)\neq g_D(a)=g_D(b)$; since $g_0(c)$ differs from $g_0(a)$ or from $g_0(b)$, one of the pairs $(c,a)$, $(c,b)$ works.

Let $\Delta(X)=\sum_{j=0}^{D}\bigl(g_j(t)-g_j(t')\bigr)X^j$. It has degree $D$ and a non-zero constant term, so $X$ is a unit of the ring $\Z_p[X]/(\Delta)$, which has $p^D$ elements. Let $N$ be the order of $X$ in this ring. Then $\Delta$ divides $X^N-1$, and $D\le N\le p^D-1$. Write
\[
S(X)=\frac{X^N-1}{\Delta(X)}=\sum_{k=0}^{L-1}s_kX^k,\qquad L=N-D+1,
\]
where $s_0=-\Delta(0)^{-1}$ and $s_{L-1}$, the inverse of the leading coefficient of $\Delta$, are non-zero. Put $s_k=0$ for the other integers $k$, and $\psi_i(u)=\sum_{j=0}^{D}s_{i-j}g_j(u)$ for $0\le i\le N$. The coefficient of $X^i$ in $S\Delta=X^N-1$ is $\psi_i(t)-\psi_i(t')$, so $\psi_i(t)=\psi_i(t')$ for $0<i<N$. Moreover $\psi_0=s_0g_0$ and $\psi_N=s_{L-1}g_D$. Hence no $\psi_i$ is a permutation.

Now consider
\[
\Sigma=\sum_{x\in\Z_p^{N+1}}\omega^{\sum_{k=0}^{L-1}s_kf(x_k,\dots,x_{k+D})}.
\]
Since $\sum_{k}s_kf(x_k,\dots,x_{k+D})=\sum_{i=0}^{N}\psi_i(x_i)$, we get $\Sigma=\prod_{i=0}^{N}\sum_{u\in\Z_p}\omega^{\psi_i(u)}$, which is non-zero by Lemma~\ref{lem:characters}. If every word of length $L$ had exactly $p^D$ local preimages, grouping the terms of $\Sigma$ according to the image word $y$ would give $\Sigma=p^D\sum_{y\in\Z_p^L}\omega^{\sum_ks_ky_k}=0$, since $s_0\neq0$. This contradiction proves the first assertion, with $L\le p^D-D$, and $F$ is not surjective by Theorem~\ref{theo: hedlund surj}.

For the second assertion, the maps $g_j$ of the inessential variables are constant; we add them to an essential map and delete these variables, keeping the offsets of the other ones. If only one variable is essential, $F$ applies its map at one position, so $F$ is surjective if and only if this map is a permutation. Otherwise, a permutation at the leftmost or at the rightmost essential position gives surjectivity~\cite[Proposition 5.22]{kurka2003topological}, and if there is none, the first assertion applies, after shifting the offsets. Finally, an $LR$-separated rule with $r-\ell\le2$ is a sum of maps of one variable, since $\pi$ depends on at most one variable, and so is a totally separated rule; their extremal monomials are permutations exactly when $\gcd(q_\ell,p-1)=1$ or $\gcd(q_r,p-1)=1$, respectively. \qed
\end{proof}

The case $p=3$ and $r-\ell=3$ with a Latin middle map is also covered. Here Latin means that every row and every column of $\pi:\Z_3^2\to\Z_3$ is a permutation. Every permutation of $\Z_3$ is affine, so each row has the form $\pi(x,y)=a_xy+b_x$ with $a_x\neq0$. Two rows with different slopes would agree in one column, which is impossible; thus $a_x=A$ does not depend on $x$, and $x\mapsto b_x$ is a permutation, so $b_x=Bx+c$ with $B\neq0$. Hence $\pi(x,y)=Ay+Bx+c$, the rule is totally separated, and Proposition~\ref{prop:total-surj} applies.

Therefore Conjecture~\ref{conj:characterization surjective CA Zp} remains open only when $r-\ell\ge3$, neither extremal monomial is a permutation, and $\pi$ is balanced but is not a sum of maps of one variable.

\newpage

\section{Reversibility}
\label{sec:reversibility}

This section is devoted to the study of reversibility of $LR$-separated CA.

Here reversibility is used in the classical CA sense: the global map $F$ has an inverse which is itself a CA~\cite{kari2005reversible}, so every configuration has a unique predecessor, computed by a local rule. This notion is different from other uses of reversibility. In concurrent computation, \emph{causal reversibility} allows an action to be undone once its consequences have been undone; Bernardo et al.~\cite{bernardo2023causal} relate it to the \emph{time reversibility} of stochastic processes, as for a stationary Markov chain with transition matrix $P$ and stationary distribution $\mu$ satisfying the detailed-balance equations $\mu(s)P(s,t)=\mu(t)P(t,s)$. These equations are not part of the definition of a reversible CA, and we do not assume a reversing symmetry either, that is, an involution $R$ with $RFR=F^{-1}$.
Marin and Rossi~\cite{marin2014lumpability,marin2017lumpability} study how time reversal interacts with \emph{lumpability}, the aggregation of states of a Markov chain which keeps the Markov property. The reduction to the essential neighborhood in the proof of Theorem~\ref{thm:characterization surjective CA} is different: it keeps the configuration space and gives the same global map. The reduction modulo a divisor of $m$ in Lemma~\ref{lem:reduction} below is a factor map, closer to lumpability, but our proofs do not use results on lumpability.

\begin{lemma}
\label{cor:permutive implies not injective}
Let $F$ be a $LR$-separated CA over $\Z_{m}$, for any integer $m \geq 3$, and let $\ell$ \resp{$r$} be the leftmost \resp{rightmost} position of $F$ with $\ell<r$.\\ 
If $F$ is $\ell$-permutive and $r$-permutive, then $F$ is not injective.
\end{lemma}

\begin{remark}
    The authors are aware that the validity of the result can be deduced as an implication of the fact that bipermutive CA whose neighborhood contains both negative and positive offsets are positively expansive, and therefore cannot be injective.
    Nevertheless, we decide to include an alternative proof, as it offers a constructive approach that highlights the dynamics involved.
\end{remark}
     
\begin{proof}
Let $F$ be a $LR$-separated CA with diameter $d$ and local rule $f$ over $\Z_m$. Following the idea of Remark \ref{re: def sep}, $f$ can be written as follows:
$$f(x_{1},\dots,x_{d+1})=a_{\ell}x_{\ell}^{q_\ell} + \pi(x_{\ell +1},\dots,x_{r-1}) + a_{r}x_{r}^{q_r}.$$
Since $F$ is $\ell$-permutive and $r$-permutive, the maps $g(x)=a_{\ell}x^{q_{\ell}}$ and  $h(x)=a_r x^{q_r}$ are invertible. 
Hence, there is $b_{\ell}, c_{r}\in\Z_{m}\setminus\{0\}$ such that $g(b_\ell)\equiv 1 \mod m$ and $h(c_{r})\equiv -1 \mod m$.
Let us denote $a=\pi(0,\cdots, 0)$. Since $g$ and $h$ are invertible, one can construct two configurations $y\neq y'$ such that: 
\begin{itemize}
    \item[$\bullet$] $y_{\cc{\ell}{r}}=b_\ell 0\cdots0 c_r$ and $y'_{\cc{\ell}{r}}=00\cdots00.$
    \item[$\bullet$] For all $i>r$: the letter $y_i$ \resp{$y'_i$} is the solution of the equation $h(x)+\pi(y_{i-r+\ell+1},\cdots,y_{i-1})+g(y_{i-r+\ell})=a$, \resp{$h(x)+\pi(y'_{i-r+\ell+1},\cdots,y'_{i-1})+g(y'_{i-r+\ell})=a$}.
    \item[$\bullet$] For all $i<\ell$: the letter $y_i$ \resp{$y'_i$} is the solution of the equation $g(x)+\pi(y_{i+1},\cdots,y_{i+r-\ell-1})+h(y_{i+r-\ell})=a$, \resp{$g(x)+\pi(y'_{i+1},\cdots,y'_{i+r-\ell-1})+h(y'_{i+r-\ell})=a$}.
\end{itemize}
By construction, $F(y)=F(y')=a^\infty$. Hence, $F$ is not injective since $y\neq y'$. \qed
\end{proof}

The proof only uses the permutivity of $f$ at positions $\ell$ and $r$, so it applies to any CA over $\Z_m$, $m\ge2$, whose local rule is permutive at its leftmost and rightmost essential positions $\ell<r$. The same construction shows that every configuration has exactly $m^{r-\ell}$ preimages: a preimage is determined by any $r-\ell$ consecutive letters, which can be chosen freely, since right permutivity then gives the letters on the right one by one, and left permutivity those on the left. For a constant configuration $a^\infty$, the map which sends a block of $r-\ell$ consecutive letters of a preimage to the next block is a permutation of $\Z_m^{r-\ell}$, so all preimages of $a^\infty$ are periodic, and two of them have a common period.

Before characterizing injectivity, we give a bound on the period of two periodic configurations with the same image when the local rule is a sum of maps of one variable; it also uses Lemma~\ref{lem:characters}. For $n\ge1$, let $\mathrm{Per}_n(\Z_m)=\{x\in\Z_m^\Z:\sigma^n(x)=x\}$ be the finite set of configurations of period $n$; every CA maps it to itself.

\begin{proposition}\label{prop:sum}
Let $p$ be prime and $d\ge1$, and let $F$ be a CA over $\Z_p$ whose local rule, with neighborhood $\cc{0}{d}$, is a sum $f(x_0,\dots,x_d)=\sum_{j=0}^{d}g_j(x_j)$ of maps $g_j:\Z_p\to\Z_p$. If at least two maps $g_j$ are not constant, then $F$ is not injective on $\mathrm{Per}_n(\Z_p)$ for some $n$ with $1\le n\le p^d-1$. The same holds for any other neighborhood of $d+1$ consecutive cells.
\end{proposition}

\begin{proof}
Choose two maps $g_a,g_b$ which are not constant. As in the proof of Proposition~\ref{prop:total-surj}, there are $t\neq t'$ with $g_a(t)\neq g_a(t')$ and $g_b(t)\neq g_b(t')$. Put $\Delta_j=g_j(t)-g_j(t')$, and let $v<w$ be the smallest and the largest $j$ with $\Delta_j\neq0$. The polynomial $Q(X)=\sum_{j=v}^{w}\Delta_jX^{j-v}$ has positive degree and a non-zero constant term, so $X$ is a unit of the ring $\Z_p[X]/(Q)$, which has $p^{w-v}$ elements. Let $n$ be the order of $X$ in this ring; then $Q$ divides $X^n-1$, and $n\le p^{w-v}-1\le p^d-1$. Write $S(X)=(X^n-1)/Q(X)=\sum_{i=0}^{n-1}s_iX^i$; the vector $s=(s_0,\dots,s_{n-1})$ is not zero. Since $S(X)\,X^vQ(X)=X^v(X^n-1)$ is zero modulo $X^n-1$, comparing coefficients, with indices modulo $n$, gives
\begin{equation}\label{eq:kernel}
\sum_{j=0}^{d}s_{i-j}\Delta_j=0\qquad\text{for every } i\in\Z_n.
\end{equation}
Assume that $F$ is injective on $\mathrm{Per}_n(\Z_p)$. We identify a configuration of period $n$ with the word $x=(x_0,\dots,x_{n-1})$, indices modulo $n$, so that $F$ is a permutation of $\Z_p^n$. With $\langle s,y\rangle=\sum_is_iy_i$, this gives
\[
\Sigma=\sum_{x\in\Z_p^n}\omega^{\langle s,F(x)\rangle}=\sum_{y\in\Z_p^n}\omega^{\langle s,y\rangle}=0,
\]
since some $s_i$ is non-zero. On the other hand, grouping the terms by input coordinate gives
\[
\Sigma=\prod_{i\in\Z_n}\sum_{u\in\Z_p}\omega^{\psi_i(u)},\qquad \psi_i(u)=\sum_{j=0}^{d}s_{i-j}g_j(u).
\]
By~\eqref{eq:kernel}, $\psi_i(t)=\psi_i(t')$ for every $i$, so no $\psi_i$ is a permutation, and every factor is non-zero by Lemma~\ref{lem:characters}. This contradicts $\Sigma=0$. For another neighborhood, $F$ is composed with a power of $\sigma$, which maps $\mathrm{Per}_n(\Z_p)$ onto itself. \qed
\end{proof}

To characterize injectivity we need two more tools. The first one reduces a question over $\Z_m$ to a question over $\Z_b$, for a divisor $b$ of $m$.

\begin{lemma}\label{lem:reduction}
\begin{enumerate}
\item If a CA $F$ over $\Z_m$ is not injective, then two distinct spatially periodic configurations with a common period have the same image.
\item Let $b\ge2$ divide $m$, let $F$ be a CA over $\Z_m$ whose local rule is a polynomial with coefficients in $\Z_m$, and let $\overline F$ be the CA over $\Z_b$ obtained by reducing the local rule modulo $b$. If $\overline F$ is not injective, then $F$ is not injective.
\end{enumerate}
\end{lemma}

\begin{proof}
For the first assertion, use the neighborhood $\cc{0}{D}$. If $D=0$, two distinct letters with the same image give two distinct constant configurations with the same image. If $D\ge1$, consider the paired de Bruijn graph~\cite{sutner1991bruijn}: its vertices are the pairs $(u,v)$ of words of length $D$, and there is an edge from $(u,v)$ to $(u',v')$ when $u'$ and $v'$ are obtained from $uc$ and $ve$ by deleting the first letter, for letters $c,e$ with $f(uc)=f(ve)$. Two configurations $x\neq y$ with the same image, $F(x)=F(y)$, give a bi-infinite path through a vertex $(u,v)$ with $u\neq v$, and every closed walk gives two periodic configurations with a common period and the same image. If such an off-diagonal vertex appears twice on the path, the part between the two visits is a closed walk. Otherwise the path stays on the diagonal $\{(u,u)\}$ before some position and after some position; the diagonal is a copy of the de Bruijn graph, which is strongly connected, so the finite part outside the diagonal can be closed by a path on the diagonal.

For the second assertion, let $R:\mathrm{Per}_n(\Z_m)\to\mathrm{Per}_n(\Z_b)$ be the reduction modulo $b$ in every cell. It is onto, and $R\circ F=\overline F\circ R$ since the local rule is a polynomial. By the first assertion, $\overline F$ is not injective on $\mathrm{Per}_n(\Z_b)$ for some $n$. If $F$ were injective, it would be onto the finite set $\mathrm{Per}_n(\Z_m)$; then $\overline F\circ R=R\circ F$ would be onto $\mathrm{Per}_n(\Z_b)$, so $\overline F$ would be onto, hence injective, on this finite set, a contradiction. \qed
\end{proof}

In the second assertion, the polynomial assumption is only used to get $R\circ F=\overline F\circ R$; it can be replaced by the assumption that the local rule respects congruence modulo $b$. An arbitrary map over a composite ring does not always do so.

The last tool is an obstruction for an additively separated extremal variable. For a map $h:\Z_m\to\Z_m$ and $c\in\Z_m$, let $S_h(c)=\sum_{t\in\Z_m}e^{2\pi i\,c\,h(t)/m}$.

\begin{proposition}\label{prop:endpoint}
Let $D\ge1$, and let $F$ be a CA over $\Z_m$ whose local rule, with neighborhood $\cc{0}{D}$, is
\[
f(x_0,\dots,x_D)=\alpha(x_0,\dots,x_{D-1})+h(x_D),
\]
where $\alpha$ is any map and $h:\Z_m\to\Z_m$ is not constant and not a permutation. Let $I=h(\Z_m)$ and $I-I=\{s-s':s,s'\in I\}$. Then $F$ is not injective if one of the following conditions holds:
\begin{enumerate}
\item $I-I=\Z_m$;
\item $I-I$ generates the additive group $\Z_m$, and $S_h(c)\neq0$ for every $c\neq0$;
\item $I-I$ generates the additive group $\Z_m$, and $\alpha(x_0,\dots,x_{D-1})=g(x_0)+\beta(x_1,\dots,x_{D-1})$, where $g$ is a permutation of $\Z_m$.
\end{enumerate}
In particular, if $m=p$ is prime, then $F$ is not injective, for every map $\alpha$. The same holds, with the conditions reflected, for a local rule $f=h(x_0)+\alpha(x_1,\dots,x_D)$.
\end{proposition}

\begin{proof}
Suppose that $F$ is injective. Then $F$ is surjective and $F^{-1}$ is a CA, so every word of length $n\ge1$ has exactly $m^D$ local preimages, of length $n+D$ (Theorem~\ref{theo: hedlund surj}). Here, for a word $w$ on $\cc{a}{b}$, a local preimage is a word $P$ on $\cc{a}{b+D}$ with $f(P_i,\dots,P_{i+D})=w_i$ for $a\le i\le b$.

Fix $a\le-D-1$ and $b_0\ge0$ such that $F^{-1}(y)_0$ depends only on $y_{\cc{a}{b_0}}$. For $b\ge-D$, let $E(b)$ be the property: for every word $w$ on $\cc{a}{b}$, all local preimages of $w$ have the same letter at position $0$. If $b\ge b_0$, then $E(b)$ holds: a local preimage $P$ of $w$ extends to a configuration $x$ with $F(x)$ equal to $w$ on $\cc{a}{b}$, so $P_0=F^{-1}(F(x))_0$ depends only on $w$. Also, $E(b)$ implies $E(b+1)$. But $E(-D)$ fails: take $t\neq t'$ with $h(t)=h(t')$, fix any letters at positions $a,\dots,-1$, and put $t$ or $t'$ at position $0$; the outputs on $\cc{a}{-D}$ are the same, since position $0$ is used only by the output at $-D$, through $h$.

Let $b>-D$ be the smallest integer with $E(b)$, and let $w$ be a word on $\cc{a}{b-1}$ whose local preimages do not all have the same letter at position $0$. For such a preimage $P$, let $u(P)$ be its last $D$ letters, and for $z,s\in\Z_m$ let $N(z)$ be the number of local preimages $P$ of $w$ with $\alpha(u(P))=z$, and $H(s)=\#h^{-1}(s)$. The local preimages of the word $wv$, for $v\in\Z_m$, are the words $Pc$ with $\alpha(u(P))+h(c)=v$, so by balance
\[
\sum_{z\in\Z_m}N(z)H(v-z)=m^D\qquad\text{for every } v\in\Z_m.
\]
For $v\in\Z_m$, every preimage $P$ with $\alpha(u(P))\in v-I$ extends to a preimage $Pc$ of $wv$, so by $E(b)$ all these $P$ have the same letter at position $0$. Taking $v=z+s$ with $s\in I$, all $P$ with $\alpha(u(P))=z$ have the same letter; and if $z'-z=s-s'$ with $s,s'\in I$, taking $v=z+s=z'+s'$ shows that $z$ and $z'$ give the same letter. So two values $z,z'$ with $N(z)>0$ and $N(z')>0$ give the same letter whenever $z'-z\in I-I$.

Under condition~1, this shows that all preimages of $w$ have the same letter at position $0$, a contradiction. Under conditions~2 and~3, we show below that $N(z)=m^{D-1}>0$ for every $z\in\Z_m$. Since $I-I$ generates $\Z_m$ and $-(I-I)=I-I$, any two values $z,z'$ are joined by a chain of values in which two consecutive values differ by an element of $I-I$; so again all preimages of $w$ have the same letter at position $0$, a contradiction.

Under condition~2, let $\zeta=e^{2\pi i/m}$ and $\widehat N(c)=\sum_zN(z)\zeta^{cz}$. The left side of the equation above is a convolution, so for $c\neq0$ we get $\widehat N(c)S_h(c)=0$; hence $\widehat N(c)=0$ for $c\neq0$, and $N$ is constant. As its values add up to $m^D$, the number of preimages of $w$, it is equal to $m^{D-1}$. Under condition~3, the local rule is permutive at position $0$, so a preimage $P$ of $w$ is determined by its last $D$ letters, which can be chosen freely: the other letters are then found from right to left. Since $g$ is a permutation, $\alpha$ takes each value on exactly $m^{D-1}$ words of length $D$, so again $N(z)=m^{D-1}$.

Finally, let $m=p$ be prime. For $c\neq0$, the map $t\mapsto ch(t)$ is not a permutation, so $S_h(c)\neq0$ by Lemma~\ref{lem:characters}. Since $h$ is not constant, $I-I$ contains a non-zero element, which generates $\Z_p$. So condition~2 holds. The case $f=h(x_0)+\alpha(x_1,\dots,x_D)$ follows by reflecting the lattice, $x_i\mapsto x_{-i}$. \qed
\end{proof}


\begin{theorem}
\label{thm:characterization injective/bijective}
    Let $F$ be a $LR$-separated CA with diameter $d=2\rho$ and local rule $f$ over $\Z_{m}$, for any integer $m \geq 3$, and let $\ell$ \resp{$r$} be the leftmost \resp{rightmost} position of $F$. If $\ell=r$, then $F$ is injective if and only if $\Cmon{m}{q_\ell}$ holds. If $\ell<r$, let $h_\ell(t)=a_\ell t^{q_\ell}$ and $h_r(t)=a_rt^{q_r}$; then $F$ is not injective if one of the following conditions holds:
    \begin{enumerate}
    \item $m$ is prime;
    \item the middle map $\pi$ is a polynomial map;
    \item $h_\ell$ or $h_r$ is a permutation;
    \item $h(\Z_m)-h(\Z_m)=\Z_m$, for $h=h_\ell$ or for $h=h_r$;
    \item $S_h(c)\neq0$ for every $c\neq0$, for $h=h_\ell$ or for $h=h_r$.
    \end{enumerate}
    Consequently, if $\ell=r$ or one of these conditions holds, then $F$ is injective if and only if $\ell = r$ and $\Cmon{m}{q_\ell}$ holds. This is the case, in particular, when $m$ is prime or $\pi$ is a polynomial map.
\end{theorem}

\begin{proof}
    We will prove the two implications separately.
    \begin{itemize}
        \item[$\Leftarrow$] If $\ell = r$, then $F$ applies the map $t\mapsto a_{\ell}t^{q_\ell}+c$ at the position of $x_\ell$, and since $\Cmon{m}{q_\ell}$ holds this map is a permutation by Lemma~\ref{lem:permutive CA iff (n,phi(m)) coprime}, so the injectivity of $F$ follows directly.
        \item[$\Rightarrow$] Suppose $F$ is injective. If $\ell=r$, the map $t\mapsto a_{\ell}t^{q_\ell}+c$ is injective (otherwise two distinct constant configurations would have the same image), so $\Cmon{m}{q_\ell}$ holds by Lemma~\ref{lem:permutive CA iff (n,phi(m)) coprime}. Assume by contradiction that $\ell < r$, and let $h_\ell(t)=a_\ell t^{q_\ell}$ and $h_r(t)=a_rt^{q_r}$; these maps are not constant, since $h_\ell(1)-h_\ell(0)=a_\ell$ and $h_r(1)-h_r(0)=a_r$ are units. In particular, $h(\Z_m)-h(\Z_m)$ contains a unit, so it generates $\Z_m$, for $h=h_\ell$ and for $h=h_r$. As in the proof of Theorem~\ref{thm:characterization surjective CA}, we delete the inessential variables and, changing all offsets at once, we use the neighborhood $\cc{0}{D}$ with $D=r-\ell$; this does not change injectivity.

        First let $m=p$ be prime. If $h_r$ is not a permutation, Proposition~\ref{prop:endpoint}, applied with $h=h_r$ and $\alpha=h_\ell(x_0)+\pi(x_1,\dots,x_{D-1})$, shows that $F$ is not injective; if $h_\ell$ is not a permutation, we use the second case of Proposition~\ref{prop:endpoint}. If both are permutations, $F$ is not injective by Lemma~\ref{cor:permutive implies not injective}.

        Now let $\pi$ be a polynomial map, and let $p$ be a prime divisor of $m$. Reducing $f$ modulo $p$ gives a CA $\overline F$ over $\Z_p$ whose extremal coefficients are non-zero, since $a_\ell$ and $a_r$ are units. If $p=2$, both extremal maps of $\overline F$ are the identity of $\Z_2$, so $\overline F$ is not injective by Lemma~\ref{cor:permutive implies not injective} and the remark after its proof. If $p$ is odd, $\overline F$ is $LR$-separated over $\Z_p$ with the same $\ell<r$, and it is not injective by the prime case. In both cases Lemma~\ref{lem:reduction} shows that $F$ is not injective, contradicting the hypothesis.

        Under condition~3, if $h_\ell$ and $h_r$ are both permutations, $F$ is not injective by Lemma~\ref{cor:permutive implies not injective}. Otherwise one of them, say $h_r$, is not a permutation while the other one is; then Proposition~\ref{prop:endpoint} applies to $h=h_r$ by its condition~3, with $g=h_\ell$ and $\beta=\pi$ (for $h=h_\ell$ we use the reflected case). Under condition~4, if $h_\ell$ or $h_r$ is a permutation we are in condition~3; otherwise Proposition~\ref{prop:endpoint} applies to $h$ by its condition~1. Under condition~5, $h$ is not a permutation, since $S_h(c)=\sum_{t\in\Z_m}e^{2\pi i\,ct/m}=0$ for $c\neq0$ when $h$ is a permutation; so Proposition~\ref{prop:endpoint} applies by its condition~2. In all cases $F$ is not injective, a contradiction. \qed
    \end{itemize}

\end{proof}
\begin{remark}
As in the case of Theorem~\ref{thm:characterization surjective CA}, it follows from Theorem~\ref{thm:characterization injective/bijective} that, when $m$ is prime, $F$ is injective if and only if $\ell = r$ and $\gcd(q_{\ell},m-1)=1$, for every middle map $\pi$. For a composite $m$ and a non-polynomial $\pi$, Theorem~\ref{thm:characterization injective/bijective} does not decide the question only when, for $h=h_\ell$ and for $h=h_r$, the map $h$ is not a permutation, $h(\Z_m)-h(\Z_m)\neq\Z_m$, and $S_h(c)=0$ for some $c\neq0$; this happens, for example, for $h(t)=t^2$ on $\Z_4$, since $h(\Z_4)=\{0,1\}$ and $S_h(2)=0$. Theorem~\ref{thm:characterization injective/bijective} also covers mixed middle terms: the polynomial rule $x_1^2+x_2x_3+x_4^2$ over $\Z_4$ or $\Z_9$ is not injective, although its extremal maps are not permutations; over $\Z_9$, for example, $(41)^\infty$ and $1^\infty$ have the same image $3^\infty$. Similarly, over $\Z_5$ every rule $x_1^4+\pi(x_2,x_3)+bx_4^4$ with $b\neq0$ is not injective, for every map $\pi$.
\end{remark}
\begin{corollary}
    Let $F$ be a $LR$-separated CA over $\Z_{m}$, where $m$ is an integer with $m \geq 3$, such that $\ell=r$ or one of the conditions of Theorem~\ref{thm:characterization injective/bijective} holds; for instance, $m$ is prime or $\pi$ is a polynomial map.
    Then $F$ is bijective if and only if $\ell=r$ and $\Cmon{m}{q_\ell}$ holds. In this case, if $F(x)_i=h(x_{i+k})$ with $h(t)=a_\ell t^{q_\ell}+c$, then $F^{-1}(y)_i=h^{-1}(y_{i-k})$.
\end{corollary}

The next two examples show that the injectivity of the global map is implied neither by the permutivity of the local rule in one variable, nor by the bijectivity of the polynomial obtained by giving the same value to all the variables.

\begin{example}
Let $F$ be a CA with local rule:
$f(a,b,c)=a^4+3b \mod 7$, with the neighborhood $\cc{-1}{1}$.
The global rule $F$ is not injective since $F((56)^\infty)=F((43)^\infty)=(26)^\infty.$
However, $P(x)=x^4+3x \mod 7$, is a permutation polynomial over $\Z_7$.
\end{example}

\begin{example}
Let $F$ be a CA with local rule:
$f(a,b,c)=a^3+2b+c^2 \mod 5$, with the neighborhood $\cc{-1}{1}$.
The global rule $F$ is not injective since $F((10)^\infty)=F((3)^\infty)=2^\infty$.
However, $P(x)=x^3+2x+x^2 \mod 5$, is a permutation polynomial over $\Z_5$ (even though it is the sum of two non permutation polynomials $P_1(x)=x^3+2x \mod 5$ and $P_2(x)=x^2 \mod 5$).
\end{example}

\begin{example}\label{ex:outside}
The hypotheses of Theorem~\ref{thm:characterization injective/bijective} cannot be removed. With the neighborhood $\cc{-1}{1}$, the rules $f_1(a,b,c)=a+2b+2$ and $f_2(a,b,c)=2ab+2bc+3a+2b$ over $\Z_4$ define reversible CA, with inverses
\[
F_1^{-1}(y)_i=y_{i+1}+2y_{i+2}+2,\qquad F_2^{-1}(y)_i=3y_{i+1}+2y_{i+2}(y_{i+1}+y_{i+3}+1).
\]
The first formula follows by direct substitution. For $F_2$, the output at $i$ is $x_{i-1}$ modulo $2$, and multiplying the output at $i+1$ by $3$ gives $x_i=3y_{i+1}+2x_{i+1}(x_i+x_{i+2}+1)$, where only parities matter in the last term. These rules are not $LR$-separated: in $f_1$ the rightmost essential variable $b$ has the coefficient $2$, which is not a unit of $\Z_4$, and in $f_2$ the rightmost essential variable $c$ appears only in the mixed term $2bc$. Additive separation is needed also over a prime field: on $\Z_3$, let $\theta_0$ exchange $1$ and $2$ and fix $0$, and let $\theta_1=\theta_2$ be the identity. The CA $T(x)_i=\theta_{x_{i+1}}(x_i)$ has two essential variables and is an involution: $T$ does not change the set of cells in state $0$, so a second application exchanges $1$ and $2$ again at the same cells. Its local rule $f(a,b)=\theta_b(a)$ is not additively separated in any variable. Two distinct periodic configurations with the same image prove that a CA is not injective, but testing finitely many periods cannot prove injectivity; this is why we give the inverses.
\end{example}

\section{Conclusions and Future Directions}

In this work, we analyzed the structural properties of non-linear CA, focusing on permutivity, surjectivity, and reversibility. We introduced the class of $j$-separated non-linear CA and established algebraic characterizations of the above mentioned properties for this class of CA.

Our findings show that permutivity plays a central role in determining surjectivity and reversibility. Specifically, we provided a sufficient condition for a $LR$-separated non-linear CA to be surjective, and we proved that over a prime field this condition is also necessary when the local rule is a sum of maps of one variable, which settles Conjecture~\ref{conj:characterization surjective CA Zp} for all radius-one rules. Additionally, we proved that, over a prime field or when the middle map is a polynomial, reversibility is equivalent to the CA being surjective with the local rule $f$ depending only on one variable; Theorem~\ref{thm:characterization injective/bijective} also covers other middle maps over composite rings. Two problems remain open: Conjecture~\ref{conj:characterization surjective CA Zp} when $r-\ell\ge3$ and $\pi$ is not a sum of maps of one variable, and reversibility in the composite cases not covered by Theorem~\ref{thm:characterization injective/bijective}. These results contribute to a deeper understanding of non-linear CA dynamics and provide a framework for identifying their computational potential.

Beyond theoretical results, we presented illustrative examples to clarify the interplay between permutivity, surjectivity, and reversibility.
\\
Although our analysis (Theorem~\ref{thm:characterization injective/bijective} in particular) shows that the class of $LR$-separated non-linear CA is not especially suited for cryptographic applications, it is the authors' belief that future investigation of broader classes of non-linear CA should lead to a complete exploitation of the non-linear rules computational complexity to produce secure cryptographic primitives.
One structural feature which allows non-linear interaction together with reversibility is a second state component at each cell, as in the classical second-order construction~\cite{toffoli1987machines}: if $H$ is any CA on $B^\Z$, where $B$ is a finite abelian group, the map $T(u,v)=(v,H(v)-u)$ on $(B\times B)^\Z$ is a reversible CA, with inverse $T^{-1}(u,v)=(H(u)-v,u)$. Of course, neither non-linearity nor reversibility alone gives cryptographic security.

We conclude by proposing some questions, related to the above discussion, that we find particularly interesting and worth exploring:
\begin{enumerate}
    \item In the case of finite fields it holds that every non-linear function is a polynomial, thus restricting significantly the possible non-linearity structures. We noticed in Remark~\ref{rem:function over finite ring not necessarily poly} that this does not hold, for example, for arbitrary finite rings: what happens in the case of a general alphabet $A$?
    \item In this work we focus on uniform CA, meaning all local interactions are determined by the same rule.
    How do our results transform in the case of non-uniform CA (i.e. a CA allowing different local rules)?
    \item For stochastic CA, reversibility can mean time reversibility with respect to a stationary measure. Which local conditions ensure it, and how does it relate to lumpability~\cite{marin2014lumpability,marin2017lumpability} and causal reversibility~\cite{bernardo2023causal}?
\end{enumerate}

\end{sloppypar}

\bibliographystyle{alpha}
\bibliography{bibliography}

\end{document}